\documentclass[12pt]{article}
\usepackage{graphicx,amssymb,amsmath,amsthm,url}

\newtheorem{theorem}{Theorem}
\newtheorem{lemma}[theorem]{Lemma}
\newtheorem{construction}{Construction}
\newtheorem{conjecture}{Conjecture}

\begin{document}
\title{Non-overlapping codes}
\author{Simon R.\ Blackburn\\
Department of Mathematics\\
Royal Holloway, University of London\\
Egham, Surrey TW20 0EX\\
United Kingdom}
\maketitle

\begin{abstract}
  We say that a $q$-ary length $n$ code is \emph{non-overlapping} if
  the set of non-trivial prefixes of codewords and the set of
  non-trivial suffices of codewords are disjoint. These codes were
  first studied by Levenshtein in 1964, motivated by applications in
  synchronisation. More recently these codes were independently
  invented (under the name \emph{cross-bifix-free} codes) by Baji\'c
  and Stojanovi\'c.

  We provide a simple construction for a class of non-overlapping
  codes which has optimal cardinality whenever $n$
  divides~$q$. Moreover, for all parameters $n$ and $q$ we show that a
  code from this class is close to optimal, in the sense that it has
  cardinality within a constant factor of an upper bound due to
  Levenshtein from 1970. Previous constructions have cardinality
  within a constant factor of the upper bound only when $q$ is fixed.

  Chee, Kiah, Purkayastha and Wang showed that a $q$-ary length~$n$
  non-overlapping code contains at most $q^n/(2n-1)$ codewords; this
  bound is weaker than the Levenshtein bound. Their proof appealed to
  the application in synchronisation: we provide a direct combinatorial argument to
  establish the bound of Chee \emph{et al}.

  We also consider codes of short length, finding the leading term of
  the maximal cardinality of a non-overlapping code when $n$ is fixed
  and $q\rightarrow \infty$. The largest cardinality of
  non-overlapping codes of lengths~$3$ or less is determined exactly.
\end{abstract}

\section{Introduction}
\label{sec:introduction}

Let $u$ and $v$ be two words (not necessarily distinct) of length $n$,
over a finite alphabet $F$ of cardinality $q$. We say that $u$ and $v$
are \emph{overlapping} if a non-empty proper prefix of $u$ is equal to
a non-empty proper suffix of $v$, or if a non-empty proper prefix of
$v$ is equal to a non-empty proper suffix of $u$. So, for example, the
binary words $00000$ and $01111$ are overlapping; so are the words
$10001$ and $11110$. However, the words $11111$ and $01110$ are
non-overlapping.

We say that a code $C\subseteq F^n$ is \emph{non-overlapping} if for
all (not necessarily distinct) $u,v\in C$,  the words $u$ and $v$ are
non-overlapping. The following is an example of a non-overlapping
binary code of length $6$ containing 3 codewords:
\[
C=\{001101,001011,001111\}.
\]
We write $C(n,q)$ for the maximum number of codewords in a
$q$-ary non-overlapping code of length $n$. It is easy to see that
$C(1,q)=q$. From now on, to avoid trivialities, we always assume that
$n\geq 2$. 

Non-overlapping codes were introduced by
Levenshtein~\cite{Levenshtein64} in 1964 (under the name `strongly
regular code'; in later papers~\cite{Levenshtein70,Levenshtein04} he
refers to `codes without overlaps'). These codes are interesting for
synchronisation applications: they are comma-free codes with the
strong property that an error in a codeword or in a state of a certain
decoding automaton does not propagate into incorrect decoding of
subsequent codewords.

Inspired by the use of distributed sequences in frame synchronisation
applications by van Wijngaarden and Willink~\cite{WijngaardenWillink},
Baji\'c and Stojanovi\'c~\cite{BajicStojanovic} recently independently
rediscovered non-overlapping codes (using the term cross-bifix-free). See
also~\cite{Bajic,BajicStojanovicLindner,BilottaGrazzini,BilottaPergola,CheeKiah12,WijngaardenWillink}
for recent papers studying non-overlapping (cross-bifix-free) codes and
their applications to synchronisation.

Levenshtein~\cite{Levenshtein70,Levenshtein04} provides a construction
for non-overlapping codes that has good performance when $q=2$, and
attributes this class of codes to Gilbert~\cite{Gilbert60}; see
Construction~\ref{con:Chee} in Section~\ref{sec:construction}
below. Chee, Kiah, Purkayastha and Wang~\cite{CheeKiah12} rediscovered
this construction, and verified by computer search that it was optimal
(in the sense of producing non-overlapping codes of largest possible
cardinality) for $q=2$ and $n\leq 16$, except when
$n=9$. Levenshtein~\cite{Levenshtein04} suggests that the main
question in the area is to prove or contradict the question of whether
this construction always produces optimal non-overlapping codes. This
question is still open for binary codes (for $n\not=9$). The main aim
of this paper is to provide a strongly negative answer to this
question, by giving a simple generalisation of the construction (see
Construction~\ref{con:new} in Section~\ref{sec:construction} below)
that performs much better when $q$ is large. This new construction is
almost optimal in the sense that its cardinality is within a constant
factor of an upper bound on non-overlapping codes due to
Levenshtein~\cite{Levenshtein70}. Previously, constructions with this
property were only known when $q$ is fixed. When $n$ divides $q$ the
cardinality of the new construction meets Levenshtein's upper bound
and so is optimal.

A second aim of the paper is to simplify some of the arguments in the literature on non-overlapping codes, either by providing purely combinatorial arguments that do not rely on knowledge of any particular application of these codes, or by providing (possibly weaker) bounds on code size that are substantially easier to prove. As an example of the former, we reprove an upper bound on $C(n,q)$ due to Chee,
Kiah, Purkayastha and Wang~\cite{CheeKiah12} using a simple combinatorial argument; their proof required an understanding of Bajic et al.'s analysis~\cite{BajicStojanovicLindner} of the variance of synchronisation time when a non-overlapping code is used in a particular application. As an example of the latter, we provide a lower bound on $C(n,q)$ when $q$ is fixed and $n\rightarrow\infty$ that is weaker than a bound due to Gilbert~\cite{Gilbert60} and Levenshtein~\cite{Levenshtein64}, but avoids the need to know any analytic combinatorics.

The remainder of the paper is structured as follows.

In Section~\ref{sec:upper}, we recap two upper bounds on the
cardinality of a non-overlapping code. The first bound, due to Chee,
Kiah, Purkayastha and Wang~\cite{CheeKiah12} states that $C(n,q)\leq
q^n/(2n-1)$. As mentioned above, Chee \emph{et al.}\ established their bound by appealing
to the application in synchronisation (deriving the bound from the
fact that a certain variance must be positive). We provide a direct
combinatorial proof of this bound. (Indeed, the combinatorial
derivation allows us to improve the bound slightly to a strict
inequality.) The upper bound due to Levenshtein~\cite{Levenshtein70}
is always better than the bound due to Chee \emph{et al.}, but
requires a little analytic combinatorics: we include this bound and
its beautiful proof for completeness.

In Section~\ref{sec:construction} we turn to constructions for
non-overlapping codes. We describe the construction due to Levenshtein
(rediscovered by Chee \emph{et al.})  and provide a simple
argument to show that this construction has cardinality within a
constant of Levenshtein's bound when $q$ is fixed. We then describe a
generalisation of the construction that performs better when $q$ is
large, and show that the codes that are produced are optimal when $n$ divides $q$.

In Section~\ref{sec:small_length} we consider non-overlapping codes of
small length. We determine $C(n,q)$ when $n\leq 3$ exactly, and we
calculate $\lim_{q\rightarrow\infty} C(n,q)/q^n$ for any fixed $q$.

Finally, in Section~\ref{sec:general_parameters}, we show that the new
construction in Section~\ref{sec:construction} produces codes that are
almost optimal for all parameters $n$ and $q$.

\section{Two upper bounds}
\label{sec:upper}

We first provide a direct combinatorial proof of the following theorem, that
slightly strengthens the bound due to Chee \emph{et al.}~\cite{CheeKiah12}.

\begin{theorem}
\label{thm:upper}
  Let $n$ and $q$ be integers with $n\geq 2$ and $q\geq 2$. Let
  $C(n,q)$ be the number of codewords in the largest non-overlapping
  $q$-ary code of length~$n$. Then
\[
C(n,q)<\frac{q^n}{2n-1}.
\]
\end{theorem}
\begin{proof}
  Let $C$ be a non-overlapping code of length $n$ over an alphabet $F$
  with $|F|=q$. Consider the set $X$ of pairs $(w,i)$ where $w\in
  F^{2n-1}$, $i\in\{1,2,\ldots,2n-1\}$ and the (cyclic) subword of $w$
  starting at position $i$ lies in $C$. So, for example, if $C$ is
  the code in the introduction then $(01111110011,8)\in X$.

  We see that $|X|=(2n-1)|C|q^{n-1}$, since there are $2n-1$ choices
  for $i$, then $|C|$ choices for the codeword starting in the $i$th
  position of $w$, then $q^{n-1}$ choices for the remaining positions
  in $w$.

  Since $C$ is non-overlapping, two codewords cannot appear as
  distinct cyclic subwords of any word $w$ of length $2n-1$. Thus,
  for any $w\in F^{2n-1}$ there is at most one choice for an integer $i$ such
  that $(w,i)\in X$. Moreover, no subword of any of the $q$ constant
  words $w$ of length $2n-1$ can appear as a codeword in a
  non-overlapping code. So $|X|\leq q^{2n-1}-q<q^{2n-1}$.

  The theorem now follows from the inequality
\[
(2n-1)|C|q^{n-1}\leq|X|< q^{2n-1}.\qedhere
\]
\end{proof}

The above theorem has the advantage of an elementary proof. The
following bound, due to Levenshtein~\cite{Levenshtein70}, is
stronger that the bound above, but requires knowledge of some analytic
combinatorics. The proof is short and beautiful, so it is included for
completeness.

\begin{theorem}
\label{thm:Levupper}
Let $n$ and $q$ be integers with $n\geq 2$ and $q\geq 2$. Let
  $C(n,q)$ be the number of codewords in the largest non-overlapping
  $q$-ary code of length~$n$. Then
\[
C(n,q)\leq\frac{1}{n}\left(\frac{n-1}{n}\right)^{n-1}q^n.
\]
\end{theorem}
\begin{proof}
Let $C$ be a $q$-ary non-overlapping code of length $n$. We say that a (finite)
$q$-ary word $s$ is \emph{$C$-free} if no subword of $s$ lies in
$C$. In other words, $s$ is $C$-free if $s$ cannot be written in the form $ucv$ where
$c\in C$, and $u$ and $v$ are (possibly empty) $q$-ary words. Note
that all sequences of length $i$ with $i<n$ are $C$-free.

We write $b_i$ for the number of $C$-free $q$-ary sequences of
length~$i$. We have $b_i=q^i$ when $i<n$.

Let $i\geq n$. Let $P$ be the set of $q$-ary words of length $i$ that
begin with a $C$-free word of length $i-1$. Let $Q$ be the set of
$C$-free words of length $i$. Let $T$ be the set of words beginning
with a $C$-free word of length $i-n$, and ending in a codeword. We
have $|P|=qb_{i-1}$, $|Q|=b_i$ and $|T|=|C|b_{i-n}$. For any code $C$,
$P\setminus Q\subseteq T$. However, the fact that $C$ is non-overlapping
implies that $P\setminus Q=T$. Thus $qb_{i-1}-b_i= |C|b_{i-n}$. This
recurrence, together with the facts that $b_0=1$ and $b_i=qb_{i-1}$
for $i<n$, imply that
\[
\sum_{i=0}^\infty b_iz^i = \frac{1}{1-qz+|C|z^n}.
\]

The radius of convergence $R$ of the power series above is finite, and
$R^{-1}$ is equal to the largest modulus of a root of the polynomial
$f(z):=1-qz+|C|z^n$. Since the coefficients $b_i$ in the power series
are all non-negative, Pringsheim's Theorem (see~\cite[Theorem IV.6,
page 240]{FlajoletSedgewick}, for example) implies that $R^{-1}$ is a root
of $f$. In particular, $f$ has a root on the positive real
axis.

We note that $f(0)>0$, and a short calculation shows that $f(z)$ has a unique
minima at $z=z_0$, where
\[
z_0=\left(\frac{q}{n|C|}\right)^{1/(n-1)}.
\]
Since $f$ has a root on the positive real axis, we must have
$f(z_0)\leq 0$. More explicitly, we find that
\[
1-q
\left(\frac{q}{n|C|}\right)^{1/(n-1)}+|C|\left(\frac{q}{n|C|}\right)^{n/(n-1)}\leq
0
\]
and so
\[
|C|^{1/(n-1)}-q\left(\frac{q}{n}\right)^{1/(n-1)}+\frac{q}{n}\left(\frac{q}{n}\right)^{1/(n-1)}\leq 0.
\]
Rearranging this inequality gives us the bound we require.
\end{proof}

\section{Constructions of non-overlapping codes}
\label{sec:construction}

Let $F=\{0,1,\ldots, q-1\}$. The following class of non-overlapping
codes of length $n$ over $F$ was proposed by
Levenshtein~\cite{Levenshtein64,Levenshtein70};
Gilbert~\cite{Gilbert60} also considered this class of codes in the context of
synchronisation applications. The codes were recently rediscovered by Chee \emph{et al.}~\cite{CheeKiah12}.

\begin{construction}[Levenshtein~\cite{Levenshtein64,Levenshtein70};
Gilbert~\cite{Gilbert60}; Chee \emph{et al.}~\cite{CheeKiah12}]
\label{con:Chee}
Let $k$ be an integer such that $1\leq k\leq n-1$. Let $C$ be the set of all
words $c\in F^n$ such that:
\begin{itemize}
\item $c_i=0$ for $1\leq i\leq k$ (so all codewords start with $k$
  zeroes);
\item $c_{k+1}\not=0$, and $c_n\not=0$;
\item the sequence $c_{k+2},c_{k+3},\ldots,c_{n-1}$ does not contain
  $k$ consecutive zeroes.
\end{itemize}
Then $C$ is a non-overlapping code. 
\end{construction}

The binary non-overlapping code $C$ of length $6$ given in the introduction is an instance of Construction~\ref{con:Chee} with $k=2$.

It is not hard to see that the construction above is indeed a
non-overlapping code. Chee \emph{et al.}\  show that the construction is
already good for small parameters. Indeed, they show that for binary
codes, Construction~\ref{con:Chee} (with the best choice of $k$)
achieves the best possible code size whenever $n\leq 16$ and $n\neq 9$.

It less clear how to choose $k$ in general so that $C$ is
as large as possible, and what the resulting asymptotic size of the
code is. However, Gilbert~\cite{Gilbert60} and
Levenshtein~\cite{Levenshtein64} show that when $q$ is fixed, and $k$
is chosen appropriately (as a function of $n$), we have that
\[
|C|\gtrsim \frac{q-1}{qe} \frac{q^n}{n}
\]
where $e$ is the base of the natural logarithm, and
$n\rightarrow\infty$ over the subsequence $n=(q^i-1)/(q-1)$. (See also Chee
\emph{et al.}~\cite{CheeKiah12}.) This lower bound on $C(n,q)$ shows that
Theorems~\ref{thm:upper} and~\ref{thm:Levupper} are tight to within a
constant factor when $q$ is fixed. The following elementary
lemma is also sufficient to establish this. The lemma is included because its proof is simpler than Levenshtein's lower bound: it avoids the use of any analytic combinatorics. Note however that the lower bound of the lemma is substantially weaker than the bound of Levenshtein~\cite{Levenshtein64} if $q$ is allowed to grow. 

\begin{lemma}
\label{lem:easy}
Let $q$ be a fixed integer, $q\geq 2$. Then the codes in
Construction~\ref{con:Chee} show that
\[
\liminf_{n\rightarrow\infty}C(n,q)/(q^n/n)\geq \frac{(q-1)^2(2q-1)}{4q^4},
\]
\end{lemma}
\begin{proof}
We begin by claiming that when $2k\leq n-2$ the number of $q$-ary
sequences of length $n-k-2$ containing no $k$ consecutive zeros is at
least
\[
q^{n-k-2}-(n-2k-1)q^{n-2k-2}.
\]
To see this, note that any sequence that fails the condition of
containing no $k$ consecutive sequences of zeroes must contain $k$
consecutive zeros starting at some position $i$, where $1\leq i\leq
n-k-2-(k-1)$. Since there are $n-2k-1$ possibilities for $i$, and
$q^{n-2k-2}$ sequences containing $k$ zeros starting at position $i$,
our claim follows. Thus, if $C$ is the non-overlapping code in
Construction~\ref{con:Chee},
\[
|C|\geq
(q-1)^2(q^{n-k-2}-nq^{n-2k-2})=\left(\frac{q-1}{q}\right)^2q^n(q^{-k}-nq^{-2k}).
\]
The function $q^{-k}-nq^{-2k}$ is maximised when
$k=\log_q(2n)+\delta$, where $\delta$ is chosen so that $|\delta|<1$
and $k$ is an integer. In this case, the value of $q^{-k}-nq^{-2k}$ is
bounded below by $(2q-1)/(4nq^2)$ (this can be shown by always taking
$\delta$ to be non-negative). Thus
\[
|C|\geq \left(\frac{(q-1)^2(2q-1)}{4nq^4}\right)q^n.\qedhere
\]
\end{proof}

When the alphabet size $q$ is much larger than the length $n$,
Construction~\ref{con:Chee} produces codes that are much smaller than
the upper bound in Theorem~\ref{thm:Levupper}. The following
generalisation of Construction~\ref{con:Chee} does not have this
drawback; indeed the construction often produces optimal
non-overlapping codes. We discuss this issue further later in this
section, and in Sections~\ref{sec:small_length}
and~\ref{sec:general_parameters} below.

Let $S\subseteq F^k$. As in the proof of Theorem~\ref{thm:Levupper},
we say that a word $x_1x_2\cdots x_r\in F^r$ is \emph{$S$-free} if
$r<k$, or if $r\geq k$ and $x_ix_{i+1}\cdots x_{i+k-1}\not\in S$ for
all $i\in\{1,2,\ldots,r-k+1\}$.

\begin{construction}
\label{con:new}
Let $k$ and $\ell$ be such that $1\leq k\leq n-1$ and $1\leq \ell\leq
q-1$. Let $F=I\cup J$ be a partition of a set $F$ of cardinality $q$
into two parts $I$ and $J$ of cardinalities $\ell$ and $q-\ell$
respectively. Let $S\subseteq I^k\subseteq F^k$. Let $C$ be the set of all words $c\in F^n$ such that:
\begin{itemize}
\item $c_1c_2\cdots c_k\in S$;
\item $c_{k+1}\in J$, and $c_n\in J$;
\item the word $c_{k+2},c_{k+3},\ldots,c_{n-1}$ is $S$-free.
\end{itemize}
Then $C$ is a non-overlapping code. 
\end{construction}

For example, suppose that $n=6$, $\ell=2$, $F=I\cup J=\{0,1\}\cup\{2\}$, $k=2$ and $S=\{00,01,10\}$. Then
\begin{align*}
C=\{&002022,002112,002122,002122,002202,002212,\\
&012022,012112,012122,012122,012202,012212,\\
&102022,102112,102122,102122,102202,102212\}.
\end{align*}

It is easy to see that Construction~\ref{con:Chee} is the special case
of Construction~\ref{con:new} with $\ell=1$, $I=\{0\}$ and
$S=\{0^k\}$.

The case of Construction~\ref{con:new} when $k=n-1$ and
$S=I^k$ is of special interest, as it produces optimal codes for an
infinite collection of parameters. In this case, $C$ is the set of words of
length $n$ whose first $n-1$ components lie in $I$, and whose final
component lies in $J$. If $n$ divides $q$, we may choose $I$ and $J$
to be such that $|I|=((n-1)/n)q$ and $|J|=(1/n)q$, so
\[
|C|=|I|^{n-1}|J|=\frac{1}{n}\left(\frac{n-1}{n}\right)^{n-1}q^n.
\]
Combining this observation with Theorem~\ref{thm:Levupper}, we see that
the following theorem holds.
\begin{theorem}
\label{thm:good_parameters}
Let $n$ and $q$ be positive integers such that $n\geq 2$ and $q\geq 2$. Let the largest non-overlapping code have cardinality
$C(n,q)$. When $n$ divides $q$,
\[
C(n,q)=\frac{1}{n}\left(\frac{n-1}{n}\right)^{n-1}q^n.
\]
Moreover, Construction~\ref{con:new} provides codes of cardinality $C(n,q)$.
\end{theorem}

It seems surprising that such a simple construction produces optimal non-overlapping codes for a wide range of parameters. One way of providing some intuition about this is as follows.  When we heavily restrict the allowed short suffixes of a code (by mandating that last symbol of each codeword lies in a small set $J$), we can make the code non-overlapping by imposing a very mild restriction on the majority of each codeword (that the first $n-1$ symbols lie in a large set $I$). However, optimal non-overlapping codes are very large, and so we cannot restrict the suffices that appear in such a code much: this means $J$ cannot be too small.

\section{Non-overlapping codes of small length}
\label{sec:small_length}

This section considers non-overlapping codes of fixed length $n$, when the
alphabet size $q$ becomes large. In this situation,
Construction~\ref{con:Chee} produces codes that are much smaller than
the upper bound in Theorem~\ref{thm:Levupper}. To see this, note that
there are at most $q^{n-k}$ codewords in a code $C$ from
Construction~\ref{con:Chee}, since the first $k$ components of any
codeword are fixed. So, since $k$ is positive, $|C|\leq q^{n-1}$ and
therefore $|C|/(q^n/n)\leq n/q$.

We saw in Section~\ref{sec:construction} that the codes given by
Construction~\ref{con:new} are optimal whenever $n$ divides $q$. The next theorem shows that the
Construction~\ref{con:new} is close to optimal when $q$ is large, even
when $n$ does not divide $q$.

\begin{theorem}
\label{thm:n_fixed}
Let $n$ be a fixed positive integer, $n\geq 2$. Then
\[
\liminf_{q\rightarrow\infty} C(n,q)/q^n=\frac{1}{n}
\left(\frac{n-1}{n}\right)^{n-1}.
\]
\end{theorem}
\begin{proof}
The upper bound follows from Theorem~\ref{thm:Levupper}. For the lower bound,
  we use Construction~\ref{con:new} in the special case when
  $k=n-1$ and $S=I^k$. In this case (in the notation of
  Construction~\ref{con:new}) $C$ is the set of words whose first $n-1$
  components lie in $I$, and whose final component lies in $J$. So here
  $|C|=\ell^{n-1}(q-\ell)$.

Let $\ell=\lceil ((n-1)/n)q\rceil$. Since $q-\ell\geq (1/n)q-1$, we
find that
\[
|C|=\frac{1}{n}\left(\frac{n-1}{n}\right)^{n-1}q^n-O(q^{n-1}),
\]
and so the theorem follows.
\end{proof}

The following two theorems provide the precise values of $C(n,q)$ when
$n=2$ and $n=3$ respectively.

\begin{theorem}
\label{thm:n_2}
A largest $q$-ary length $2$ non-overlapping code has $C(2,q)$
codewords, where $C(2,q)=\lfloor q/2\rfloor\, \lceil q/2\rceil$.
\end{theorem}
\begin{proof}
Construction~\ref{con:new} in the case $n=2$, $k=1$, $\ell=\lfloor
q/2\rfloor$ and $S=I^k$ provides the lower bound on $C(2,q)$ we require.

Let $C$ be a $q$-ary non-overlapping code of length $q$. Let $I$ be the set
of symbols which occur in the first position of a codeword in $C$, and
let $J$ be the set of symbols that occur in the final position of a
codeword in $C$. Since $C$ is non-overlapping, $I$ and $J$ are
disjoint. Thus
\[
 |C|\leq |I||J|\leq |I|(q-|I|)\leq \lfloor q/2 \rfloor \lceil
 q/2\rceil.
\]
\end{proof}

In the following theorem, $[x]$ denotes the nearest integer to the
real number $x$.

\begin{theorem}
\label{thm:n_3}
A largest $q$-ary length $3$ non-overlapping code has $C(3,q)$
codewords, where $C(3,q)=[2q/3]^2(q-[2q/3])$.
\end{theorem}
\begin{proof}
Construction~\ref{con:new} in the case $n=3$, $k=2$, $\ell=[
2q/3]$ and $S=I^k$ provides the lower bound on $C(2,q)$ we require.

Let $C$ be a $q$-ary non-overlapping code of length $q$ of maximal
size. Let $F$ be the underlying alphabet of $C$, so $|F|=q$.

Let $I$ be the set of symbols which occur in the first position
of a codeword in $C$. Let $J$ be the complement of $I$ in
$F$, so $|J|=q-|I|$. Since $C$ is non-overlapping,
the symbols that occur in the final component of any codeword lie in
$J$. So we may write $C$ as a disjoint union $C=C_1\cup C_2$, where
$C_1\subseteq I\times I\times J$ and $C_2\subseteq I\times J\times J$.

Let $X$ be the set of all pairs $(b,c)\in I\times J$ such
that $abc\in C$ for some $a\in I$. Define
\begin{align*}
\overline{C_1}&=\{abc\mid a\in I\text{ and }(b,c)\in X\},\\
\overline{C_2}&=\{bcd\mid (b,c)\in (I\times J)\setminus X\text{ and } d\in J\}.
\end{align*}
Clearly $C_1\subseteq \overline{C_1}$. Moreover, $C_2\subseteq
\overline{C_2}$, since whenever $bcd\in C$ is a codeword, the fact
that $C$ is non-overlapping implies that
$(b,c)\not\in X$. But $\overline{C}=\overline{C_1}\cup\overline{C_2}$
is a non-overlapping code, and so $C=\overline{C}$ as $C$ is maximal.

We have
\[
|C|=|\overline{C}|=|X| |I|+(|I| |J| - |X|)|J|=|X|(|I|-|J|)+|I||J|^2.
\]
If $|I|\leq |J|$, then the maximum value of $|C|$ is achieved when
$|X|=0$, at $\max_{i\in\{1,2,\ldots,\lfloor q/2\rfloor\}}i^2(q-i)$. If
$|I|>|J|$, the maximum value of $|C|$ is achieved when $|X|=|I||J|$,
at $\max_{i\in\{\lfloor q/2\rfloor,\lfloor
  q/2\rfloor+1,\ldots,q-1\}}i^2(q-i)$. Thus
\[
|C|\leq \max_{i\in\{1,2,\ldots,q-1\}}i^2(q-i)= [2q/3]^2(q-[2q/3]),
\]
and so the theorem follows.
\end{proof}

It would be interesting to know if the following conjecture is true:

\begin{conjecture}
\label{conj_strong}
Let $n$ be an integer such that $n\geq 2$. For all sufficiently large
integers $q$, a largest $q$-ary non-overlapping code of length $n$ is
given by Construction~\ref{con:new} in the case $k=n-1$ (and some
value of $\ell$).
\end{conjecture}

\section{Good constructions for general parameters}
\label{sec:general_parameters}

This section shows that Construction~\ref{con:new} is always good, in
the sense that it produces non-overlapping codes of cardinality within a
constant factor of the upper bound given by Theorem~\ref{thm:Levupper}
for all parameters. This is implied by the proof of the theorem
below. Up to now, we have only used the special case of
Construction~\ref{con:new} when $S=I^k$. However, in this section we
require more general sets $S$ to avoid `rounding errors' for some sets
of parameters. We mention that the issue of rounding errors also arises in work due to Guibas and Odlyzko~\cite{GuibasOdlyzko} on prefix-synchonized codes. The class of prefix-synchronized codes is not the same as non-overlapping codes: the codes of Construction~\ref{con:Chee}, but not all codes of Construction~\ref{con:new}, are prefix-synchronised. Nevertheless, optimal prefix-synchronized codes are also large (close in size to $q^n/n$) and rounding errors cause interesting behaviour in the constructions of such codes: see Theorem~2 of~\cite{GuibasOdlyzko} and the discussion following its statement.

\begin{theorem}
\label{thm:general}
There exist absolute constants $c_1$ and $c_2$ such that
\[
c_1(q^n/n)\leq C(n,q) \leq c_2(q^n/n)
\]
for all integers $n$ and $q$ with $n\geq 2$ and $q\geq 2$.
\end{theorem}
\begin{proof}
  The existence of $c_2$ follows by the upper bound on $C(n,q)$ given
  by Theorem~\ref{thm:Levupper}. Indeed, Theorem~\ref{thm:Levupper} shows that we may take $c_2=\frac{1}{2}$. (If we are only interested in codes of large length then $c_2$ may be taken to be close to $1/e$, where $e$ is the base of the natural logarithm.) We prove the lower bound by showing that
  there exists a constant $c_1$ such that for all choices of $n$ and
  $q$, one of the constructions given by Construction~\ref{con:new}
  contains at least $c_1(q^n/n)$ codewords.

Let $(n_1,q_1),(n_2,q_2),\ldots$ be an infinite sequence of pairs of
integers where $n_i\geq 2$ and $q_i\geq 2$. It suffices to show that
$C(n_i,q_i)/(q_i^{n_i}/n_i)$ is always bounded below by some positive
constant as $i\rightarrow\infty$. Suppose, for a contradiction, that
this is not the case. By passing to a suitable subsequence if
necessary, we may assume that $C(n_i,q_i)/(q_i^{n_i}/n_i)\rightarrow 0$
as $i\rightarrow\infty$.  If the integers $q_i$ are bounded, then
Lemma~\ref{lem:easy} gives a contradiction. If the integers~$n_i$ are
bounded, we again have a contradiction, by
Theorem~\ref{thm:n_fixed}. So we may assume, without loss of
generality, that the integer sequences $(n_i)$ and $(q_i)$ are
unbounded. By passing to a suitable subsequence if necessary, we may
therefore assume that $(n_i)$ and $(q_i) $ are strictly increasing
sequences (and that $n_i$ and $q_i$ are sufficiently large for our
purposes below). In particular, we may assume that
$n_i\rightarrow\infty$ and $q_i\rightarrow\infty$ as
$i\rightarrow\infty$.

Let $k_i=\lceil \log_2 2n_i\rceil$, and set $s_i=\lfloor
q_i^{k_i}/(2n_i)\rfloor$. Let $F_i$ be a set of size $q_i$. Let
$I_i\subseteq F_i$ have cardinality $\ell_i$, where $\ell_i=\lceil
s_i^{1/k_i}\rceil$. Let $J_i$ be the complement of $I_i$ in $F_i$. Let
$S_i$ be a subset of $I_i^{k_i}$ of cardinality $s_i$. Note that such
a set $S_i$ exists, by our choice of $\ell_i$.

Let $C_i$ be the $q_i$-ary non-overlapping code of length $n_i$ given
by Construction~\ref{con:new} in the case $k=k_i$, $\ell=\ell_i$,
$I=I_i$, $J=J_i$ and $S=S_i$. Then
\begin{equation}
\label{eqn:C}
|C_i|=|S|(q_i-\ell_i)^2f_i
\end{equation}
where $f_i$ is the number of $S$-free sequences of length
$n_i-k_i-2$. We now aim to find a lower bound on $|C_i|$.

Since $q_i\rightarrow\infty$ as $i\rightarrow\infty$, we
see that
\[
q_i^{k_i}/(2n_i)\geq
q_i^{\log_2(2n_i)}/(2n_i)=2^{(\log_2(q_i)-1)(\log_2(2n_i)}\rightarrow\infty.
\]
Hence
\begin{equation}
\label{eqn:S}
|S|\sim q_i^{k_i}/(2n_i)
\end{equation}
as $i\rightarrow\infty$.

Note that
\[
(2n_i)^{(1/k_i)}\geq 2^{\log_2 (2n_i)/2\log_2(2n_i)}=2^{1/2},
\]
and hence
\[
s_i^{1/k_i}\leq \left(\frac{q_i^{k_i}}{2n_i}\right)^{1/k_i}
\leq 2^{-1/2}q_i.
\]
Since $(1-2^{-1/2})^2>(1/12)$, we see that
\begin{equation}
\label{eqn:ell}
(q_i-\ell_i)^2>(1/12)q_i^2
\end{equation}
for all sufficiently large $i$.

The number of $S$-free $q$-ary sequences of length $r$ is at least
$q^r-(r-k+1)|S|q^{r-k}$, since every word that is not $S$-free
must contain an element of $S$ somewhere as a subword. So the
number of $S$-free $q$-ary sequences of length $r$ is at least
$q^r-r|S|q^{r-k}=q^{r}(1-r|S|q^{-k})$.  Thus
\begin{equation}
\label{eqn:f}
\begin{split}
f_i&\geq q_i^{n_i-k_i-2}(1-(n_i-k_i-2)|S_i|q_i^{-k_i}\\
&\geq \frac{1}{2}q_i^{n_i-k_i-2}(2-2n_i|S_i|q_i^{-k_i})\\
&\sim \frac{1}{2}q_i^{n_i-k_i-2},
\end{split}
\end{equation}
the last step following from~\eqref{eqn:S}.

Now \eqref{eqn:S},~\eqref{eqn:ell} and~\eqref{eqn:f} combine with
\eqref{eqn:C} to show that $|C_i|> (1/50)(q_i^{n_i}/n_i)$ for all
sufficiently large $i$. This
contradiction completes the proof of the theorem.
\end{proof}

\paragraph{Acknowledgement} The author thanks the Associate Editor and the reviewers of this paper for their useful comments that have improved this paper.

\end{document}